\documentclass[aps, pra, reprint]{revtex4-1}
\pdfoutput=1
\usepackage{amsmath}
\usepackage{amssymb}
\usepackage{graphicx}
\usepackage{enumerate}
\usepackage{enumitem}

\providecommand{\U}[1]{\protect\rule{.1in}{.1in}}
\newtheorem{theorem}{Theorem}

\newtheorem{conjecture}[theorem]{Conjecture}
\newtheorem{corollary}[theorem]{Corollary}

\newtheorem{definition}[theorem]{Definition}

\newtheorem{proposition}[theorem]{Proposition}
\newtheorem{remark}[theorem]{Remark}

\newenvironment{proof}[1][Proof]{\noindent\textbf{#1.} }{\ \rule{0.5em}{0.5em}}

\begin{document}
\title{Quantum cellular automata without particles}
\author{David A. Meyer }
\email{dmeyer@math.ucsd.edu}
\affiliation{Department of Mathematics,  University of California,  San Diego, 
 La Jolla,  California 92093-0112,  USA}

\author{Asif Shakeel}
\email{ashakeel@ucsd.edu}
\affiliation{Department of Mathematics,  University of California,  San Diego, 
 La Jolla,  California 92093-0112,  USA}

\date{January 15, 2016}

\begin{abstract}
Quantum cellular automata (QCA) constitute space and time homogeneous discrete models for quantum field theories (QFTs).  Although QFTs are defined without reference to
particles, computations are done in terms of Feynman diagrams, which are
explicitly interpreted in terms of interacting particles.  Similarly, the
easiest QCA to construct are quantum lattice gas automata (QLGA).  A natural
question then is, which QCA are not  QLGA?  Here we 
construct a non-trivial  example of such a QCA; it provides a simple model in  $1+1$
dimensions with no particle interpretation at the scale where the QCA  dynamics are homogeneous.
\end{abstract}

\maketitle

\section{Introduction} \label{section:intro}
The famous talk in which Feynman proposed the idea of quantum computers was 
entitled ``Simulating physics with computers''~\cite{bib:Feynman}.  In it he
takes the point of view that to simulate physics on a computer, space and time
should be discretized, and the dynamics should be local and causal; he comments that 
a natural architecture would be a cellular automaton.  Since his goal is to 
simulate quantum physics, this would have to be a {\it quantum\/} cellular automaton
(QCA).  We may think of this as a discrete quantum field theory, discrete in
space and time, but also with only a finite dimensional Hilbert space associated
with each spacetime lattice site.  Such a system should be able to simulate, for 
example, scattering of particles in $\phi^4$ quantum field theory, as has recently 
been shown to be possible with a quantum gate array architecture~\cite{bib:JLP}.

The investigation of QCA dates back to early papers of Grossing and Zeilinger~\cite{gz:qca}, of Meyer~\cite{bib:meyer2,dm:u1dqca}, and of Durr and Santha~\cite{dls:dpwfqca}, some of which focused on the technically easier case of periodic boundary conditions. One of the main difficulties has always been to construct explicit models in which the conditions of translation invariance, unitarity, locality and causality are simultaneously satisfied.  The easiest way forward has been to reinterpret translation invariance in order to allow quantum lattice gas automata (QLGA) models.  We may  note that  studying  (classical) lattice gas automata as a representative class within  (classical) cellular automata (CA) has precedent in the works of  Toffoli~\cite{tcm:whcalg}  and H\'{e}non~\cite{Henon}.  

The more recent work, within the last decade, of Schumacher and Werner~\cite{sw:rvqca}, Arrighi, {\it et al.}~\cite {anw:odqca}, and of Shakeel and Love~\cite{bib:slwqcaqlga} has developed frameworks within which more general models satisfying the requisite conditions can be constructed.   Schumacher and Werner~\cite{sw:rvqca}, and Gross, {\it et al.}~\cite{bib:itodqwca}, have developed an axiomatic
definition of QCA based on the {\it desiderata\/} of unitarity, locality and 
causality, as well as of translation invariance (spatial homogeneity)~\cite{sw:rvqca}. In this formalism, the  QCA dynamics are described on a  {\it quasi-local} algebra, by which they mean an  increasing chain of  finitely many tensor products of finite dimensional $C^*$-algebras.  The QCA evolution is in the Heisenberg picture, and given by an automorphism of quasi-local  algebra by  {\it local rules}~\cite{sw:rvqca, bib:itodqwca}.  Within this setup, an  {\it index theory} for a classification of  one-dimensional QCA was developed in~\cite{bib:itodqwca}. A parallel  definition of one-dimensional  QCA  by  Arrighi, {\it et al.}~\cite{anw:odqca}  is given in the Schr\"{o}dinger picture, i.e., one in which a QCA is defined on a  Hilbert space and evolves unitarily. In this case the Hilbert space has as its basis  {\it finite but unbounded  configurations} (finitely many active cells in a quiescent background), and  the QCA state in it evolves  by a unitary, casual and translation-invariant {\it global  evolution} operator.   Shakeel and Love~\cite{bib:slwqcaqlga}, working in the latter formalism  and building on it, found  conditions that determine when a  multi-dimensional QCA is a  QLGA. They investigated  QLGA because they are simple models for QCA and have obvious physical interpretations~\cite{bib:meyer2}. Their characterization uses a different set of algebraic substructures than the ones used in the index theory in~\cite{bib:itodqwca}, and is motivated  by  the goal of  classification of multi-dimensional QCA explicitly by the form of the global  evolution.  They provide examples of QCA that are not QLGA, even in $1+1$ dimensions.  These 
examples, however, do not propagate information, which leaves open the question 
of the existence of  QCA in the Schr\"{o}dinger picture model which have no particle interpretation, but 
which nevertheless propagate information.  Analogously, the $\phi^4$ quantum 
field theory simulation result~\cite{bib:JLP} leaves open the question of how 
efficient a simulation is possible in regimes or for initial conditions not 
interpretable as scattering particles.  

Our goal here is to show by a simple  $(1+1)$-dimensional construction in the Schr\"{o}dinger picture that there are discrete quantum field theories (QCA) that propagate information, yet without particle interpretation (i.e., not QLGA) at  the time scale where the QCA dynamics are homogeneous.  In the Heisenberg picture,  such examples are contained in the   Clifford QCA  presented in  the work of Schlingemann, {\it et al.}~\cite{svw:oscqca}. Ours is  another step in the direction of classifying  QCA by the dynamical processes at play in the transfer of information among cells. This is  a program akin to that which has been carried out for  reversible classical CA by Kari~\cite{Kari:ocdsrca}, Toffoli~\cite{tof:ccurca}, and Toffoli and Margolus~\cite{tm:icaar}.

This paper is organized as follows. Section~\ref{section:prelim} recalls the definition of QCA and QLGA, and the condition that determines when a QCA is a QLGA. Section~\ref{sec:catqlga} begins with  the example that constitutes the main part of this paper: a one-dimensional QCA that cannot be described as a QLGA,  although it is built by  concatenating  two QLGA.  Analysis of this example leads us  to a conjecture generalizing this method of  constructing QCA to arbitrary lattice dimensions, cell Hilbert spaces and neighborhoods. Useful as this technique is,  there are QCA that  are not obtainable  in this manner, as  the last proposition of the section shows. Section~\ref{sec:conc}  is the conclusion  with a summary of the results and a discussion of the wider context in which our work stands.

\section{Preliminaries} \label{section:prelim}
  
We consider the QCA model as formulated in~\cite{bib:slwqcaqlga}. In general,   the lattice of cells is $\mathbb{Z}^n$.  For the reader who is not specifically interested in the technical details encountered in defining Hilbert spaces and operators  over infinite lattices, the lattice of cells can be taken as finite.  This does not affect  any of the interesting aspects of the example QCA  we discuss. Wherever needed,  we provide equivalent (and simpler) definitions for  finite lattices,   after the definitions for the infinite lattice.  A finite lattice is of the form $\mathbb{Z}^n/(N_1\times \cdots \times N_n)$ for some finite positive integers $N_i$, $1\leq i \leq n$, with the ends wrapping cyclically, or in other words, it is  a torus. Unless we are exclusively considering the infinite lattice, in which case we will make that explicit, we  denote a lattice  by $\mathcal{L}$.  Over each cell is an identical  finite dimensional Hilbert space $W$. A cell has a finite  neighborhood $\mathcal{E} = \{e_1,e_2,\ldots,e_r\} \subset \mathcal{L}$ of size $r$, which specifies the surrounding cells that influence its evolution. The neighborhood of cell $x \in \mathcal{L}$ is denoted by $
\mathcal{E}_x= x+\mathcal{E}.$ 
  Let $\mathcal{B}$ be an orthonormal basis  of the cell Hilbert space $W$, 
\begin{equation*}
\mathcal{B} = \{ \left| b \right\rangle \}.
\end{equation*}
 For the infinite lattice $\mathcal{L}=\mathbb{Z}^n$, basis elements of the QCA Hilbert space are constructed as  sequences consisting  of a finite region of cells in {\it active\/} states immersed in a background of cells in a fixed (unit norm) {\it quiescent\/} state $\left| q \right\rangle \in W$. Hence this basis is called the {\it finite configurations\/} basis,
denoted by  $\mathcal{C}$,
\begin{align}   \label{sofc}
\mathcal{C} = \bigg\{& \bigotimes_{x \in \mathbb{Z}^n}  \left| b^x \right\rangle :
\left| b^x \right\rangle \in \mathcal{B},     \nonumber \\ & \text{all but a  finite number of } \left| b^x \right\rangle  = \left| q \right\rangle\bigg\},
\end{align}
where  $\left| b^x \right\rangle$ is the cell basis element at $x$. 
$\mathcal{C}$ is orthonormal with respect to  the  inner product induced on it from that on $W$.   The QCA Hilbert space  is the $\ell^2$-completion of $\mathcal{C}$, and is called the  {\it Hilbert space of finite configurations},   denoted by $\mathcal{H}_\mathcal{C}$.   This definition of  the  finite configurations  basis  ensures that   it is  countable, so that the Hilbert space of finite configurations is separable  (in the topological sense).  The naming convention given here is adopted from~\cite{anw:odqca}; in the mathematics literature, this construction is called an {\it incomplete infinite tensor product\/} space, as in~\cite{ag:shrt,vn:idp}.

For a  finite lattice, we use the same terminology as for the infinite lattice, except the Hilbert space on which  the  QCA evolves,  $\mathcal{H}_\mathcal{C}$,  is the usual tensor product space,
\begin{equation*}
\mathcal{H}_\mathcal{C} = \bigotimes_{x \in \mathcal{L}} W
\end{equation*}
with  the  basis $\mathcal{C}$ given by,
\begin{equation*}
\mathcal{C}=\bigl\{ \bigotimes_{x \in \mathcal{L}} \left| b^x \right\rangle :
\left| b^x \right\rangle \in \mathcal{B}\}.
\end{equation*}

A QCA evolves by a unitary transformation on $\mathcal{H}_\mathcal{C}$, called  the {\it global evolution} which we  denote by  $\mathcal{G}$.  It is required to be

\begin{enumerate}[label=(\roman{*})]

\item \label{transinvqca} \textit{Translation invariant}:  A translation operator $\tau_z$, for some $z \in \mathcal{L}$, is defined by its action on an element $\bigotimes_{x \in \mathcal{L}} \vert b^x\rangle \in {\mathcal{C}}$:
\begin{equation*}
\tau_z:  \bigotimes_{x \in \mathcal{L}}  \left| b^x \right\rangle  \mapsto \bigotimes_{x \in \mathcal{L}} \left| b^{x+z} \right\rangle
\end{equation*}
$\mathcal{G}$ is {\it translation invariant\/} if  $\tau^{\vphantom{-1}}_z   \mathcal{G}   \tau^{-1}_z = \mathcal{G} $ for all $z \in  \mathcal{L}$. 
Note $\tau_z$ is unitary, i.e.,  $\tau^{-1}_z =  \tau^{\dag}_z$.
 
\item \label{causalqca} \textit{Causal relative to a neighborhood $\mathcal{E}$}: 
$\mathcal{G}$ is {\it causal\/} relative to a neighborhood $\mathcal{E}$ if  for every pair $\rho, \rho'$,  of   density operators on ${\mathcal{H}}_{\mathcal{C}}$, and  $x \in \mathcal{L}$, that satisfy:
\begin{equation*}
\rho|_{\mathcal{E}_x} = \rho'|_{\mathcal{E}_x}\text{,}
\end{equation*}
 the  operators $\mathcal{G} \rho \mathcal{G}^\dag, \mathcal{G}\rho' \mathcal{G}^\dag$  satisfy
\begin{equation*}
\mathcal{G} \rho \mathcal{G}^\dag |_x =  \mathcal{G} \rho' \mathcal{G}^\dag |_x
\end{equation*}
\end{enumerate}
The state of a QCA is given by a density operator on ${\mathcal{H}}_{\mathcal{C}}$.

\begin{remark}
For the rest of the discussion, the {\it neighborhood\/} for a given QCA is  the unique minimal set (under set inclusion) satisfying  the causality condition above.
\end{remark}

The concept of causality can  equally well be discussed  with respect to the  evolution, under conjugation by $\mathcal{G}$, of operators {\it local\/} upon a finite number of cells, i.e., non-identity on those cells  and  identity on all other cells. 

To define local  operators for the infinite lattice, we first look at the Hilbert space  ${\mathcal{H}}_{\mathcal{C}}$ as composed of a finite part and a countably infinite part. 
\begin{definition} \label{coDdef} Let $D \subset \mathbb{Z}^n$ be a  finite subset. Define the  set of {\it co}-$D$ configurations to be ${\mathcal{C}}_{\overline{D}} := \{ \bigotimes_{i \in \mathbb{Z}^n\setminus D} \vert {c}_i \rangle : {c}_i \in {Q},  \text{ all but finite } \vert {c}_{i} \rangle  = \left| q \right\rangle \}$.  Let the inner product on ${\mathcal{C}}_{\overline{D}} $ be induced by the inner product on $W$\, as in the case of   ${\mathcal{C}}$.  Then the {\it co}-$D$ space, denoted by  $\mathcal{H}_{{\mathcal{C}}_{\overline{D}}}$, is defined as  the  completion  of $\text{span}({\mathcal{C}}_{\overline{D}})$ under the induced $l^2$ norm.
\end{definition}

To be able to refer to operators on a finite subset of tensor factors in the infinite lattice case, we simply embed  $\bigotimes_{j \in D} \begin{rm}{End}\end{rm}(W)$ into a subalgebra of $B({\mathcal{H}}_{\mathcal{C}})$    (the algebra  of bounded linear operators on $\mathcal{H}_{\mathcal{C}}$),
\begin{align} \label{algembd}
  \iota_D:     \bigotimes_{j \in D} \begin{rm}{End}\end{rm}(W)  &\hookrightarrow B({\mathcal{H}}_{\mathcal{C}}) \\ 
  a &\mapsto a \otimes \mathbb{I}_{\overline{\mathcal{D}}} \nonumber 
\end{align}
where $a$ is an element of $\bigotimes_{j \in D} \begin{rm}{End}\end{rm}(W)$, and $\mathbb{I}_{\overline{\mathcal{D}}}$ is the identity operator on the co-$D$ space, $\mathcal{H}_{{\mathcal{C}}_{\overline{D}}}$. Through the embedding $\iota_D$~\eqref{algembd} the algebra $\bigotimes_{j \in D} \begin{rm}{End}\end{rm}(W)$ is  isomorphic to the corresponding finite dimensional  subalgebra of $B({\mathcal{H}}_{\mathcal{C}})$.  Then, for the infinite lattice, we define local operators as follows.

\begin{definition}
A linear operator $M$ on ${\mathcal{H}}_{\mathcal{C}}$ is {\it local} upon a finite subset $D \in \mathbb{Z}^n$ if it is in the  image of the  map $\iota_D$~\eqref{algembd}.
\end{definition}

For a finite lattice,  the operators local upon $D  \subset \mathcal{L}$ are
\begin{align*} \bigotimes_{j \in D} \begin{rm}{End}\end{rm}(W)  \otimes \bigotimes_{j \in \mathcal{L}\setminus D}  \mathbb{I},
\end{align*}
where $\begin{rm}{End}\end{rm}(W)$ is the set of linear operators on $W$, and $\mathbb{I}$ is the identity operator on $W$.

The next theorem, relating causality to evolution of local algebras,  will need  the {\it reflected} neighborhood, denoted by $\mathcal{V}$,
\begin{align*} 
\mathcal{V} = -\mathcal{E}
\end{align*} 
As with the neighborhood, the reflected neighborhood of cell $z$ is denoted by $\mathcal{V}_z = z+\mathcal{V}$.

The expression of causality   in this  picture is given by  the  {\it structural reversibility}  theorem due to Arrighi, {\it et al.} in~\cite{anw:odqca}.  A proof of this theorem is  in~\cite{bib:slwqcaqlga}. 
\begin{theorem}[Structural Reversibility]  \label{strucreva}
  Let $\mathcal{G} : \mathcal{H}_{\mathcal{C}} \longrightarrow \mathcal{H}_{\mathcal{C}}$ be a unitary  operator  and $\mathcal{E}$ a neighborhood.  Then the following are equivalent.
 \begin{enumerate}[label=(\roman{*})] 
\item \label{strv1} $\mathcal{G}$ is causal relative to the  neighborhood $\mathcal{E}$. 
\item \label{strv2}  For every  operator $A_z$ local upon cell $z$, $\mathcal{G}^\dag A_z \mathcal{G}$ is local upon  $\mathcal{E}_z$.
\item \label{strv3} $\mathcal{G}^\dag$ is causal relative to the  reflected neighborhood $\mathcal{V}$. 
\item \label{strv4}  For every  operator $A_z$ local upon cell $z$, $\mathcal{G} A_z \mathcal{G}^\dag$ is local upon  ${\mathcal{V}}_z$.
\end{enumerate}   
\end{theorem}

A QLGA models particles propagating on a lattice and scattering by interaction at the lattice sites. Each cell can be occupied by multiple particles, and each particle has a state which is a vector in a {\it subcell\/} Hilbert  space $W_j$.  Let us say that the internal states of particle $j$, $1 \leq j \leq d$, are  elements of the subcell Hilbert space $W_j$, so the cell Hilbert space is $W = \bigotimes_{j=1}^d W_j$. The quiescent state needed for the infinite configurations basis for an infinite lattice as in Eq.~\eqref{sofc},  $\left| q \right\rangle \in \bigotimes_{j \in \mathcal{E}} W_j$,  is a pure tensor, i.e.,  has the form
\begin{equation*} 
\left| q \right\rangle = \bigotimes_{j \in \mathcal{E}} \left| q_j \right\rangle
\end{equation*}
for some  unit norm $\left| q_j \right\rangle  \in W_j$. 

The basis $\mathcal{B}$ of the cell Hilbert space $W$ is  expressed in terms of some orthonormal bases $\mathcal{B}_j$ of $W_j$, 
\begin{equation*}
\mathcal{B} = \{ \left| b \right\rangle = \bigotimes_{j=1}^{d} \left| b_j \right\rangle :   \left| b_j \right\rangle \in  \mathcal{B}_j\}.
\end{equation*}

In the propagation stage of the evolution, a particle with internal state $j$ (or equivalently, that occupies subcell state $j$) hops to  the corresponding subcell $j$ of a designated neighboring cell that is $e_j \in \mathcal{L}$ away.  Naturally,  this requires the  collection of such neighbors to be  specified by the {\it neighborhood\/} $\mathcal{E} = \{e_1,e_2,\ldots,e_d\}\subset \mathcal{L}$, of cardinality $|\mathcal{E}| = d$. Thus the global evolution is described  by two unitary steps,
 \begin{enumerate}[label=(\roman{*})] 
\item \label{propqlg} {\it Advection},  $\sigma$, that shifts the appropriate subcell state to the corresponding neighbor,
 \begin{equation}   \label{propeqn}
 \left . \begin{array} {cccc}
 \sigma : 
  \bigotimes_{x \in \mathcal{L}}   \bigotimes_{j=1}^{d} \left| b^x_{j} \right\rangle \mapsto  \bigotimes_{x \in \mathcal{L}}  \bigotimes_{j=1}^{d} \left| b^{x+e_j}_{j} \right\rangle,
    \end{array}
  \right.
\end{equation} 
where $\left| b^x_{j} \right\rangle\in  \mathcal{B}_j$, and the cell index $x$ indexes the subcell bases elements $\left| b^x_{j} \right\rangle$.

\item \label{colqlg} {\it Scattering}, $\hat S$, which acts on each cell  by a local unitary scattering map $S: W \longrightarrow W$,
\begin{equation} \label{scateqn}
\hat S:   \bigotimes_{x \in \mathcal{L}}   \bigotimes_{j=1}^{d} \left| b^x_{j} \right\rangle \mapsto   \bigotimes_{x \in \mathcal{L}}   S(\bigotimes_{j=1}^{d} \left| b^x_{j} \right\rangle).
\end{equation}
Note that in the infinite lattice case,  $S$ must fix the quiescent state, i.e., $S (\left| q \right\rangle) = \left| q \right\rangle$.
 \end{enumerate}  
 Each time step,    the current state of  the QLGA is  mapped unitarily to the next by its {\it global evolution} $\mathcal{G}$,
  \begin{equation*}  
\mathcal{G} = \hat S \sigma
\end{equation*} 
The state of a QLGA is an element of $\mathcal{H}_\mathcal{C}$.   It is clear that a QLGA is a QCA. We observe  that the advection $\sigma$ in~\ref{propqlg} is completely specified by the neighborhood $\mathcal{E}$. We denote a QLGA by a pair $(\sigma, S)$.

Let us describe the   criterion from~\cite{bib:slwqcaqlga} for a  QCA to be a QLGA. Denote the image, under $\mathcal{G}$,  of operators localized on a single cell by
\begin{equation*} 
\mathcal{G}_z = \mathcal{G}^\dag \mathcal{A}_z \mathcal{G},
\end{equation*}
  Let us further denote by $\mathcal{D}_{z,x}$ the  following subalgebra of $\mathcal{G}_z$, $z \in \mathbb{Z}^n$.
\begin{equation*} 
\mathcal{D}_{z,x} = \mathcal{G}_z   \cap \mathcal{A}_{x}
\end{equation*}

When $z \in \mathcal{V}_x$ then $\mathcal{D}_{z,x} $  are the elements of $\mathcal{G}_z$ which are contained in $\mathcal{A}_x$, where
\begin{align*} 
\mathcal{G}_z =\mathcal{G}^\dag \mathcal{A}_z\mathcal{G} \subset  \bigotimes_{ k \in \mathcal{E}_z} \begin{rm}{End}\end{rm}(W)  
\end{align*}
and
\begin{align*}
\mathcal{A}_z =  \underbrace{\begin{rm}{End}\end{rm}(W)}_{k = z}   \otimes \bigotimes_{k \in \mathcal{E}_z\setminus \{z\}} \mathbb{I}_{k}.
\end{align*}

For a QCA to be  QLGA, (by Corollary~\ref{corS} below), it is necessary and sufficient that 
\begin{equation} \label{qlgcond}
\mathcal{A}_y = \begin{rm}{span }\end{rm}(\prod_{k \in  \mathcal{E}} \mathcal{D}_{{y-k},y}),
\end{equation}
for all $y \in \mathbb{Z}^n$.

Before proceeding to our construction in the next section, we give a summary of useful results  from~\cite{bib:slwqcaqlga}, describing the embedding of patches  $\mathcal{D}_{{x-y},x}$, originating from neighboring cells, into a cell algebra, when a QCA locally satisfies a  QLGA condition.
\begin{theorem}\label{thmS}
 Suppose that  $\mathcal{G}$ is the  global evolution of a QCA with neighborhood $\mathcal{E}$. 
Then
\begin{equation*}
\mathcal{A}_x = \begin{rm}{span }\end{rm}(\prod_{y \in  \mathcal{E}} \mathcal{D}_{{x-y},x})
\end{equation*}
if and only if there exists an isomorphism of vector spaces,
\begin{equation*}
T :  W \longrightarrow \bigotimes_{j \in \mathcal{E}} W_{j}
\end{equation*}
for some vector spaces $\{W_{j}\}_{j \in \mathcal{E}}$.  Under the isomorphism $T$, for each $y   \in \mathcal{E}$,
\begin{equation*}
\mathcal{D}_{{x-y},x} \cong  \begin{rm}{End}\end{rm}(W_y) \otimes \bigotimes_{j \in \mathcal{E}, j \neq y} \mathbb{I}_{W_j}
\end{equation*}
\end{theorem}

The  algebras $\mathcal{A}_x$  and $\mathcal{G}_x$ (the images of algebras $\mathcal{A}_x$ after one timestep of the global evolution $\mathcal{G}$) are very  simply related  to  the patches $\mathcal{D}_{{x-y},x}$ under the conditions  of the above theorem. 

\begin{corollary} \label{corS} Suppose that $\mathcal{G}$ is the  global evolution of a QCA  with neighborhood $\mathcal{E}$, and satisfies   $\mathcal{A}_x = \begin{rm}{span }\end{rm}(\prod_{y \in  \mathcal{E}} \mathcal{D}_{{x-y},x})$. Then
 
\begin{enumerate}[label=(\roman{*})] 
\item \label{cor1}  $\mathcal{A}_x  = \begin{rm}{End}\end{rm}(W)  \cong \bigotimes_{j \in \mathcal{E}} \begin{rm}{End}\end{rm}(W_j)$, for all $x \in \mathbb{Z}^n$.
\item \label{cor2}  The dimension of  $W$, $d_W$, is a product of the dimensions of $W_j$, $d_{W_j}$, i.e.,  $d_W =  \prod_{j \in \mathcal{E}} d_{W_j}$.
\item \label{cor3}  $\mathcal{G}_x = \begin{rm}{span }\end{rm}(\prod_{k \in  \mathcal{E}} \mathcal{D}_{x,{x+k}}) \cong  \bigotimes_{k \in  \mathcal{E}}\begin{rm}{End}\end{rm}(W_{k})$.
\end{enumerate}
\end{corollary}

For completeness, we also include the structure theorem  that globally characterizes a QCA as a QLGA.
\begin{theorem}\label{thmStructure}
$\mathcal{G}$ is the  global evolution of a QCA  on the Hilbert space of finite configurations $\mathcal{H}_{\mathcal{ C}}$, with  neighborhood $\mathcal{E}$,  and satisfies
 $\mathcal{A}_x = \begin{rm}{span }\end{rm}(\prod_{y \in  \mathcal{E}} \mathcal{D}_{{x-y},x})$,  if and only if
\begin{enumerate}[label=(\roman{*})]
\item \label{thmitem1}  there exists  an  isomorphism of vector spaces $T$,
\begin{equation*}
T :  W \longrightarrow \bigotimes_{j \in \mathcal{E}} W_{j}
\end{equation*}
for some vector spaces $\{W_{j}\}_{j \in \mathcal{E}}$.   Under the isomorphism $T$, for each $y   \in \mathcal{E}$,
\begin{equation*}
\mathcal{D}_{{x-y},x} \cong  \begin{rm}{End}\end{rm}(W_y) \otimes \bigotimes_{j \in \mathcal{E}, j \neq y} \mathbb{I}_{W_j}
\end{equation*}
Furthermore, $\hat W =  \bigotimes_{j \in \mathcal{E}} W_{j}$ can be given  an inner product such that $T$ is an inner product preserving isomorphism of Hilbert spaces.
\item  \label{thmitem2}  $\mathcal{G}$ is given by
\begin{equation*}
\mathcal{G}   \cong       {\hat S} \sigma, 
\end{equation*}
where   $\sigma$ is as in~\eqref{propeqn}, and $\hat S$ is as in~\eqref{scateqn} in terms of a unitary map $S$ on $\bigotimes_{j \in \mathcal{E}} W_{j}$.
\item \label{thmitem3} (For infinite lattice) $\left| \tilde q \right\rangle = T(\left| q \right\rangle)$ is a pure tensor, i.e.,  $\left| \tilde q \right\rangle = \bigotimes_{j \in \mathcal{E}} \left| \tilde q_j \right\rangle$ for some  unit norm $\left| \tilde q_j \right\rangle \in  W_{j}$, and $S$ in the definition of $ \hat S$ fixes $\left| \tilde q \right\rangle$: $S \left| \tilde q \right\rangle =  \left| \tilde q \right\rangle$.
\end{enumerate}
\end{theorem}

\section{A class of QCA that are not  QLGA} \label{sec:catqlga}
We illustrate,  with a simple example,      a class of QCA without particles. This QCA is    composed of two QLGA, and is shown in Fig.~\ref{figqbg1}.

\begin{figure}[ht]
\includegraphics[
natheight=4.315400in, 
natwidth=5.554700in, 
height=2.2in, 
width=2.4in
]
{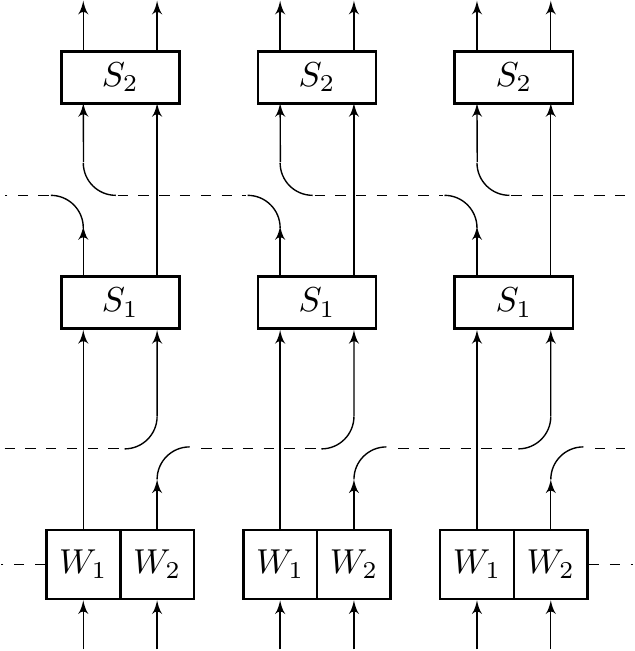}
\caption{A QCA which is not a  QLGA}
\label{figqbg1}
\end{figure}
The global evolution is
 \begin{equation}  \label{ge1}
\mathcal{G} = \hat S_2 \sigma_2  \hat S_1 \sigma_1.
\end{equation} 
Here, assume $W_1= W_2 = \mathbb{C}^2$, $\sigma_1$ is described by the neighborhood $\mathcal{E}_1=\{e_1 = 0, e_2= -1\}$ and $\sigma_2$ by the neighborhood $\mathcal{E}_2=\{e_1 = +1, e_2= 0\}$.  For the infinite lattice case, we take  the quiescent state to be  $\left| q \right\rangle = \left| 00 \right\rangle$.

When either  $S_1$ or $S_2$  is identity, this describes a QLGA. Indeed, if $S_1 = \mathbb{I}$, 
   \begin{equation*}  
\mathcal{G} = \hat S_2 \sigma_2 \sigma_1.
\end{equation*} 
For  the composite  advection $\sigma_2 \sigma_1$,  the neighborhood is $\mathcal{E}=\{e_1 = +1, e_2= -1\}$. If  $S_2 = \mathbb{I}$, 
   \begin{equation}   \label{s2id}
\mathcal{G} = \sigma_2 \hat S_1  \sigma_1.
\end{equation} 
Observe that  under the unitary isomorphism $\sigma^{-1}_2$ of $\mathcal{H}_\mathcal{C}$,
\begin{equation} \label{reperm}
\sigma^{-1}_2:  \mathcal{H}_\mathcal{C} \longrightarrow  \mathcal{H}_\mathcal{C},
\end{equation}
  we can write the global evolution as
   \begin{equation*}  
\tilde{\mathcal{G}} =   \sigma^{-1}_2  \mathcal{G} \sigma^{\vphantom{-1}}_2 = \hat S^{\vphantom{-1}}_1  \sigma^{\vphantom{-1}}_1  \sigma^{\vphantom{-1}}_2.
\end{equation*} 
This shows that $\mathcal{G}$  is equivalent to a QLGA.

We can   describe this  in a more intuitive  way, by writing the global evolution in~\eqref{s2id} as
   \begin{equation*}  
\mathcal{G}  =  \sigma^{\vphantom{-1}}_2 \hat S^{\vphantom{-1}}_1 \sigma^{-1}_2 \sigma^{\vphantom{-1}}_2 \sigma^{\vphantom{-1}}_1.
\end{equation*} 
The natural way to deal with this is to   adjust   the cell definition  to accommodate the new scattering $ \sigma^{\vphantom{-1}}_2 \hat S^{\vphantom{-1}}_1 \sigma^{-1}_2$. Using  the   tensor factor indexing as for the ``standard" cell, we can make the scattering local by redefining a  new cell  as the grouping of tensor factors from adjacent standard cells as follows. The cell at position $y$ is
   \begin{equation}   \label{newcell}
W_1^{y-1} \otimes W_2^y.
\end{equation} 
This new constructed cell is shown  in Fig.~\ref{figqbg2}, by shading the corresponding subcells in a matched way. Under this cell definition, the neighborhood for $\sigma_2 \sigma_1$ is  $\mathcal{E}=\{e_1 = +1, e_2= -1\}$. 

\begin{figure}[ht] 
\includegraphics[
natheight=4.315400in, 
natwidth=5.554700in, 
height=1.8in, 
width=3.3in
]
{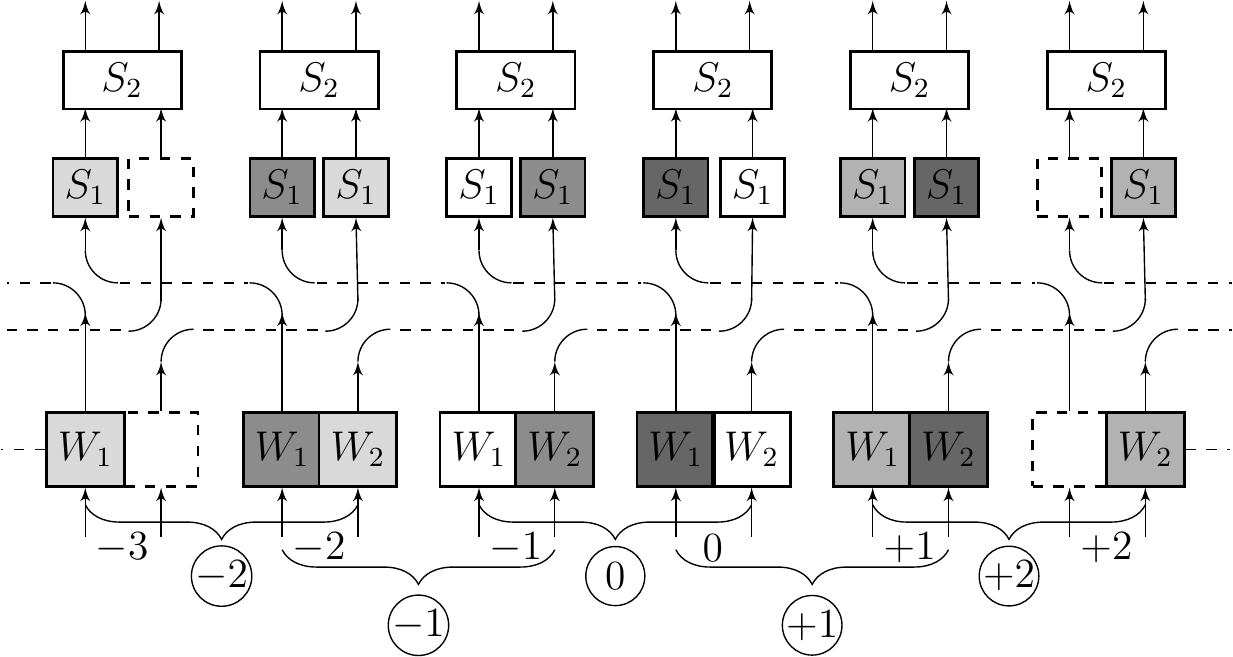}
\caption{QCA of Fig.~\ref{figqbg1} as written in Eq.~\eqref{soqca}. 
The ``regular'' cell numbers are shown without circles under each qubit 
pair. Also the cells constructed as in Eq.~\eqref{newcell} are shown as 
split across the regular cells, and their qubits shaded in a matched way, 
labeled by circled numbers.  
$\sigma_2 \hat S_1 \sigma^{-1}_2$ 
is depicted as pairs of smaller $S_1$ squares shaded corresponding to the new cell construction, where each pair of  the same shade is a single $S_1$.}
\label{figqbg2}
\end{figure}

We  define  the construction of a new  cell in a more general fashion. A cell can either be of the same size or an integer multiple of the size of the standard cell; this is dictated by the requirement that  the  evolution stay translation invariant.  A new cell of the same size as the standard  can only be defined through an advection as in the example above in~\eqref{newcell};  a set of such new cells  can be grouped to obtain a larger cell.

\begin{definition} Given a (current) cell Hilbert space $W = \bigotimes_{j=1}^d W_j$ and a set of finite configurations $\mathcal{C}$, a (new) {\it cell construction\/} is given by a neighborhood $\mathcal{E}_c = \{c_1,\ldots, c_d\}\subset\mathbb{Z}^n$ (or equivalently an advection $\sigma_c$ corresponding to $\mathcal{E}_c$) and a positive integer $m$.  Let us  represent the  current cell $x$ as a collection of pairs of indices
\begin{equation*}
 \{(x,j)\}_{j=1}^{d},
\end{equation*}
where $x$ is the cell index and $j$ is the subcell index. The new single cell $x$ is  obtained by advection to get
\begin{equation*}
 \{(x+c_j,j)\}_{j=1}^d.
\end{equation*}
A new larger cell $y$   is obtained by  grouping  $m$ such cells
\begin{equation*} 
\bigg\{ (my+c_j,j), \ldots, \bigl(m(y+1)-1+c_j,j\bigr)\bigg\}_{j=1}^d.
\end{equation*}
 In the subcell terminology, the constructed cell is 
 \begin{equation} \label{scxjindex}
 \bigotimes_{x=my}^{m(y+1)-1} \bigotimes_{j=1}^{d} W^{x+c_j}_j.
 \end{equation}

\end{definition}
We note that in~\eqref{reperm}, $\sigma_c =  \sigma^{-1}_2$ and $m=1$. 

Next  consider  the case when  neither  $S_1$ nor $S_2$  is identity. We can write 
 \begin{equation}  \label{soqca}
\mathcal{G} = \hat S^{\vphantom{-1}}_2 \sigma^{\vphantom{-1}}_2 \hat S^{\vphantom{-1}}_1 \sigma^{-1}_2 \sigma^{\vphantom{-1}}_2 \sigma^{\vphantom{-1}}_1.
\end{equation} 
This form of evolution is displayed in Fig.~\ref{figqbg2}.  
$S_1$, which now shows up as  $\sigma^{\vphantom{-1}}_2 \hat S^{\vphantom{-1}}_1 \sigma^{-1}_2$, is ``spread out''  by $\sigma_2$  over adjacent standard cells. 
However, there is no  guarantee that the  ``scattering,''
 \begin{equation*}  
\hat S^{\vphantom{-1}}_2 \sigma^{\vphantom{-1}}_2 \hat S^{\vphantom{-1}}_1 \sigma^{-1}_2,
\end{equation*}  
is  locally describable.  That is, given a pair  $S_1, S_2$, one cannot always construct  a cell  such that  $\mathcal{G}$ above can be written  as a QLGA.

\begin{theorem}
For the QCA given by \eqref{ge1}, there exist $S_1$ and $S_2$ such that the global evolution $\mathcal{G}$ in Eq.~\eqref{ge1}
is not equivalent to a QLGA for any cell construction.
\end{theorem}

\begin{proof}  
Let us first rule out the cell configurations that worked in the two cases above. For the standard cell configuration (the one that worked for $S_1=\mathbb{I}$), a choice of $S_1$ and $S_2$ can be made  to make the neighborhood  $\mathcal{E} = \{e_1 = -1, e_2 = 0, e_3 = +1\}$. For instance, we can take $S_1=S_2=S$, the symmetric scattering matrix, 
\begin{equation} \label{eq:S1}
S=\begin{pmatrix}
1&0&0&0\\
0&1/\sqrt{2}&i/\sqrt{2}&0\\
0&i/\sqrt{2}&1/\sqrt{2}&0\\
0&0&0&i\\
\end{pmatrix}.
\end{equation}
 By the criterion in~\cite{bib:slwqcaqlga} [Corollary~\ref{corS}~\ref{cor2}], the cell Hilbert space dimension, in this case $4$, must have $\left\lvert \mathcal{E} \right\rvert= r = 3$ factors to be a QLGA. This condition is clearly not met. For the case of   the cell configuration $W_1^{y-1} \otimes W_2^y$ (the one that worked when $S_2=\mathbb{I}$), the neighborhood is strictly $\mathcal{E} = \{e_1 = -2, e_2 = -1, e_3 = +1, e_4 = +2\}$, yielding  $\left\lvert \mathcal{E} \right\rvert= r = 4$. This runs into the same problem as before, as the cell Hilbert space  dimension is again $4$. Hence these cell configurations cannot correspond to QLGA. In fact any cell  construction based on single-cell ($m=1$) has this problem.

Next we consider more general cell constructions for this example. Take the cell to be $m \geq 2$ standard adjacent cells (the one that works for $S_1 =\mathbb{I}$). Then the neighborhood is $\mathcal{E} = \{e_1 = -1, e_2 = 0, e_3 = +1\}$.  $\left\lvert \mathcal{E} \right\rvert= r = 3$, and  the cell Hilbert space dimension is $2^{2m}$, so this is  in agreement with Corollary~\ref{corS}~\ref{cor2}. We need to resort to the  criterion in Eq.~\eqref{qlgcond} to show that the  assertion  still holds. Notice that $S$ in Eq.~\eqref{eq:S1} has the property that, for  $A,B,C \in M_2(\mathbb{C})$ (the algebra of $2 \times 2$ complex matrices),
\begin{equation} \label{pr2nonpr}
S^\dag (\mathbb{I} \otimes A) S = B \otimes C \iff A=B=C=\mathbb{I}.
\end{equation}
By the  symmetry of $S$, this is also true if $\mathbb{I} \otimes A$ is replaced by $A\otimes\mathbb{I}$ in the above equation.

 In the following, the subscripts and superscripts in the tensor factor indexing are as in Eq.~\eqref{scxjindex}. The property above implies, that up to conjugation by the 
local  unitary  $\bigotimes^m S_2 = \bigotimes^m S$ ({\it local\/} relative to cells, i.e., acting by a cell-wise product of unitary transformations), we have that
\begin{alignat}{3} \label{qlgstm}
&\mathcal{D}_{{y-1},y} && =  &&\bigotimes_{x=my}^{m(y+1)-1} \bigotimes_{j=1}^{2}  \mathbb{I}^{x}_j,  \nonumber \\
&\mathcal{D}_{{y},y} && =   &&\phantom{m}M_2(\mathbb{C})^{my}_1 \otimes \mathbb{I}^{my}_2 \otimes  \bigg(\bigotimes_{x=my+1}^{m(y+1)-2} \bigotimes_{j=1}^{2}  M_2(\mathbb{C})^{x}_j \bigg) 
   \nonumber  \\
& && &&\phantom{m}  \otimes \mathbb{I}^{m(y+1)-1}_1 \otimes M_2(\mathbb{C})^{m(y+1)-1}_2,   \nonumber \\
&\mathcal{D}_{{y+1},y} && =  &&\bigotimes_{x=my}^{m(y+1)-1} \bigotimes_{j=1}^{2}  \mathbb{I}^{x}_j.    \nonumber \\
\end{alignat}
This implies
\begin{equation} \label{AsupD}
\mathcal{A}_y \supsetneq \begin{rm}{span }\end{rm}(\prod_{k \in  \mathcal{E}} \mathcal{D}_{{y-k},y}).
\end{equation}
Thus this cell structure is not compatible with a QLGA.
\begin{remark} 
To attain~\eqref{AsupD}, the hypothesis on  $S$ (when $S_1=S_2=S$)  is weaker than the property in~\eqref{pr2nonpr}. Even so, the authors think $S \in M_2(\mathbb{C})$ satisfying ~\eqref{pr2nonpr} are generic. 
\end{remark}

Observe from Eq.~\eqref{qlgstm} that for  $m > 2$  the same considerations as $m=2$ apply  when the  criterion in Eq.~\eqref{qlgcond} is used.    Knowing this, we consider the  cell structure in Eq.~\eqref{newcell} and show how  a grouping of two such cells, i.e., $m=2$,  is incompatible with a QLGA when  $S_1=S_2=S$. The new ($m=2$)  cell $0$, for example,  is formed by the single cells circled $0$ and $+1$  in Fig.~\ref{figqbg2}.
 When we refer to cell $y$, we  implicitly view this larger cell $0$ as the prototype. That said, cell $y$ is
\begin{equation*}
(W^{2y-1}_1\otimes W^{2y}_2) \otimes (W^{2y}_1\otimes W^{2y+1}_2),  \\
\end{equation*}
and the same applies to  the neighbors, where  the neighborhood is  (in units of this larger cell) $\mathcal{E} = \{e_1 = -1, e_2 = 0, e_3 = +1\}$. We begin by looking at $\mathcal{D}_{{y+1},y}$. We would like to show that 
\begin{equation} \label{dimyyp1}
\text{dim } \mathcal{D}_{{y+1},y} < \text{dim } M_2(\mathbb{C}) = 4.
\end{equation}
This involves a diagram chase. First, we see that  the only influence on cell $y$ from cell $y+1$ is from the tensor factors in the following set with  the arrows showing their endpoints after advection but before the action of respective $S_1$'s,
\begin{alignat}{3} 
&W^{2y+1}_1 && \mapsto  && W^{2y}_1,     \nonumber \\
&W^{2y+2}_1 && \mapsto  && W^{2y+1}_1.     \nonumber \\
\end{alignat}
Therefore, we start with an operator which is a finite sum of elements  of the form (we include only the  identity factors  that matter)
\begin{equation} \label{stropr}
(a^{2y}_1\otimes \mathbb{I}^{2y}_2) \otimes (b^{2y+1}_1\otimes \mathbb{I}^{2y+1}_2) \otimes (\mathbb{I}^{2y+2}_1\otimes \mathbb{I}^{2y+2}_2),  
\end{equation}
where $a^{2y}_1, b^{2y+1}_1 \in M_2(\mathbb{C})$. This sum  is acted on by the relevant $S_1$'s to give a finite sum of elements of the form
\begin{equation} \label{stropr2}
\mathbb{I}^{2y}_2 \otimes S_1^\dag(a^{2y}_1 \otimes \mathbb{I}^{2y+1}_2) S_1 \otimes S_1^\dag(b^{2y+1}_1\otimes  \mathbb{I}^{2y+2}_2) S_1 \otimes \mathbb{I}^{2y+2}_1,  
\end{equation}
where we have exhibited  the indices that are acted on by the $S_1$'s. Since  $\mathcal{D}_{{y+1},y}$ is  non-identity only  on cell $y$, we see  that after the action of $S_2$'s on the above, we get a finite sum of  elements of the form
\begin{equation} \label{stropr3}
(g^{2y}_1\otimes h^{2y}_2) \otimes (\mathbb{I}^{2y+1}_1\otimes s^{2y+1}_2) \otimes (\mathbb{I}^{2y+2}_1\otimes \mathbb{I}^{2y+2}_2).  
\end{equation}
We observe that $S_2 = S$,  and  that  conjugation by local (in this case $S_2$ action on pairs of qubits) unitary operators cannot change an ``entangled'' (not a product) operator to an ``unentangled''  operator (a product).  This implies in particular that after the conjugation action of $S_1$'s (before the action of $S_2$'s), the elements in~\eqref{stropr2} are of the form
\begin{equation*} 
(t^{2y}_1\otimes \mathbb{I}^{2y}_2) \otimes (\mathbb{I}^{2y+1}_1\otimes p^{2y+1}_2) \otimes (\mathbb{I}^{2y+2}_1\otimes \mathbb{I}^{2y+2}_2).  
\end{equation*}
This further implies that the original element~\eqref{stropr}  must be  of the form
\begin{equation} \label{stropr5}
(a^{2y}_1\otimes \mathbb{I}^{2y}_2) \otimes (\mathbb{I}^{2y+1}_1\otimes \mathbb{I}^{2y+1}_2) \otimes (\mathbb{I}^{2y+2}_1\otimes \mathbb{I}^{2y+2}_2).  
\end{equation}
To show~\eqref{dimyyp1}, we need to find  an element of the  form~\eqref{stropr5} whose image after the action of $S_1$'s and $S_2$'s is not of the form~\eqref{stropr3}.  These  elements abound. For instance, using 
\begin{equation*}
a^{2y}_1 = \begin{pmatrix}
1&0\\
0&0\\
\end{pmatrix}.
\end{equation*}
in~\eqref{stropr5} works. This shows that 
\begin{equation*}
\text{dim } \mathcal{D}_{{y+1},y} < \text{dim } M_2(\mathbb{C}) = 4.
\end{equation*}
By symmetry,
\begin{equation*} 
\text{dim } \mathcal{D}_{{y-1},y} < \text{dim } M_2(\mathbb{C}) = 4.
\end{equation*}
Such reasoning  in the context of $\mathcal{D}_{{y},y}$ also shows that 
\begin{equation*} 
\text{dim } \mathcal{D}_{{y},y} = \bigl(\text{dim } M_2(\mathbb{C})\bigr)^2 = 16.
\end{equation*}
Combined, these imply
\begin{equation*}
\text{dim }\mathcal{A}_y > \text{dim }\begin{rm}{span }\end{rm}(\prod_{k \in  \mathcal{E}} \mathcal{D}_{{y-k},y}).
\end{equation*}
Thus this cell structure is also not compatible with a QLGA.  It is clear that any other cell constructions that are based on single cells  other than the ones just considered, i.e., constructed through  other advection operators, can   be similarly shown to be incompatible with a  QLGA description. For $m > 2$ the argument follows the same steps. Hence this  shows that there is a pair of $S_1$, $S_2$ such that the QCA given in~\eqref{soqca} and Fig.~\ref{figqbg1} is not a QLGA for any cell construction.
\end{proof}

We  call a neighborhood, hence a QCA, {\it trivial}, if there is only one element in the neighborhood. When concerned with a  QLGA, in which  the neighborhood determines the advection,  we also refer to   the advection as {\it trivial\/} when the neighborhood is trivial, i.e.,  $\mathcal{E}=\{e_1= \cdots = e_d\}$ (here $d$ is the number of tensor factors of the cell Hilbert space $W = \bigotimes_{j=1}^d W_j$).

Generalizing the above result, we state the following.
\begin{conjecture}
Suppose $\mathcal{H}_\mathcal{C}$ is a Hilbert space of finite configurations with the cell Hilbert space $W = \bigotimes_{j=1}^d W_j$. Let  $\sigma_1$ and $\sigma_2$ be two non-trivial advection operators on $\mathcal{H}_\mathcal{C}$. Then there exist  unitary transformations $S_1$ and $S_2$ on $W$, such that the  QCA formed by concatenating  the QLGA $(\sigma_1, S_1)$ and $(\sigma_2, S_2)$, i.e., given by the global evolution $\mathcal{G}$, 
 \begin{equation*}  
\mathcal{G} = \hat S_2 \sigma_2  \hat S_1 \sigma_1,
\end{equation*} 
is not equivalent to a QLGA for any cell construction.
\end{conjecture}

One might imagine that such constructions can yield the most general QCA. Such is not the case as shown by the following. 

\begin{proposition} \label{qcabqlga}
There exist finite length QCA that are not equivalent to a  concatenation of finitely many QLGA. 
\end{proposition}

\begin{proof}
If a QCA with cell Hilbert space of prime dimension is a concatenation of  QLGA,  then the constituent QLGA must have cell Hilbert spaces of prime dimension (the same as the QCA). Since  every   QLGA with cell Hilbert space of prime dimension is trivial, and trivial QLGA when  concatenated can only yield a trivial QCA, this implies that  the original QCA must be  trivial. But there exist non-trivial finite length QCA on qubits,  for  example some Clifford QCA (CQCA)~\cite{svw:oscqca}, in particular the CQCA  of example $1.1$ in~\cite{svw:oscqca} (such non-trivial CQCA are more generally defined in~\cite{svw:oscqca} for other prime dimensions as well).
\end{proof}
\begin{remark}
In~\cite{svw:oscqca}, the QCA definition is in the Heisenberg or operator picture. Since we are interested in unitary evolution, or the Schr\"{o}dinger picture, we may  only cite  finite length versions of the CQCA, namely those  for which the global evolution operator always exists. 
\end{remark}
This result indicates that there are interpretations of QCA beyond the regime of QLGA and concatenated constructions. The proof that we have given of the   above proposition is valid only for finite length QCA.  To the authors' knowledge,  this  proposition   is unproven for infinite length QCA as defined in this paper in the Schr\"{o}dinger picture. 

\section{Conclusion} \label{sec:conc}

In this paper, we constructed  a QCA, on a one-dimensional lattice, from a concatenation of two simple QLGA such that the constructed QCA is itself not a QLGA.  In other words this QCA has no particle interpretation at the time scale at which the QCA dynamics are homogeneous. The  proof of its non-QLGA behavior relies on application of  the condition developed in~\cite{bib:slwqcaqlga} characterizing when a QCA is a QLGA. In the same paper, it was noted that the question of complete QCA classification is still open. We hope this construction is a step in the process of answering that question. Our analysis suggests the  conjecture   that a more general result is possible, in which arbitrary cell dimension, lattice dimension, and any neighborhood scheme (except the trivial neighborhood) can be used in constructing such QCA from QLGA. For finite length QCA, we showed by citing the Clifford QCA in~\cite{svw:oscqca}, that not all QCA can be expressed as concatenations of  finitely many QLGA.    In the larger theme of quantum simulations, and in view of  the important work of Jordan,  {\it et al.}~\cite{bib:JLP} on simulation of $\phi^4$ quantum  field theory, the question arises as to the efficiency with which  a quantum  simulation model  might simulate physics without a particle description.  


\begin{acknowledgements}
This work was partially supported by AFOSR grant FA9550-12-1-0046.
\end{acknowledgements}

%


\end{document}